\documentclass[letterpaper, 10 pt, conference]{ieeeconf}

\IEEEoverridecommandlockouts
\usepackage{cite}
\usepackage{amsmath,amssymb,amsfonts}

\newtheorem{proposition}{Proposition}
\newtheorem{Lemma}{Lemma}
\newtheorem{Theorem}{Theorem}
\newtheorem{definition}{Definition}

\newtheorem{problem}{Problem}

\usepackage{array}

\usepackage{algorithmicx}
\usepackage{algpseudocode}
\usepackage{algorithm}
\algnewcommand\algorithmicforeach{\textbf{for each}}
\algdef{S}[FOR]{ForEach}[1]{\algorithmicforeach\ #1\ \algorithmicdo}

\usepackage{graphicx}
\usepackage{textcomp}
\usepackage{xcolor}
\def\BibTeX{{\rm B\kern-.05em{\sc i\kern-.025em b}\kern-.08em
    T\kern-.1667em\lower.7ex\hbox{E}\kern-.125emX}}
    
\usepackage[labelformat=simple]{subcaption}

\begin{document}

\title{Verification and Synthesis of Control Barrier Functions
}
\author{Andrew Clark \thanks{The author is with the Department of Electrical and Computer Engineering, Worcester Polytechnic Institute, 100 Institute Road, Worcester, MA, USA 01609. Email: aclark@wpi.edu.}
}

\maketitle

\begin{abstract}
Control systems often must satisfy strict safety requirements over an extended operating lifetime. Control Barrier Functions (CBFs) are a promising recent approach to constructing simple and safe control policies. This paper proposes a framework for verifying that a CBF guarantees safety for all time and synthesizing CBFs with verifiable safety in polynomial control systems. Our approach is to show that safety of CBFs is equivalent to the non-existence of solutions to a family of polynomial equations, and then prove that this non-existence is equivalent to a pair of sum-of-squares constraints via the Positivstellensatz of algebraic geometry. We develop this Positivstellensatz to verify CBFs, as well as generalization to high-degree systems and multiple CBF constraints. We then propose a set of heuristics for CBF synthesis, including a general alternating-descent heuristic, a specialized approach for compact safe regions, and an approach for convex unsafe regions. Our approach is illustrated on two numerical examples. 
\end{abstract}



\section{Introduction}
\label{sec:intro}

Safety is a critical property of cyber-physical systems in applications including autonomous vehicles \cite{koopman2017autonomous}, energy \cite{morison2004power}, and medicine \cite{rovetta2000telerobotic}. In the control-theoretic context, safety is typically defined as ensuring that the state trajectory remains within a given safe region (or, equivalently, avoids a given unsafe region) for all time. It is not always possible, however, to guarantee safety by introducing a control input, even for controllable systems without actuation constraints. This has motivated substantial research into characterizing the set of states from which safety can be guaranteed, and constructing controllers that ensure safety from those states \cite{aubin2011viability}.

Control Barrier Functions (CBFs) are a promising recent approach to safe control.  CBFs are functions of the system state that are negative outside the safe region, so that safety can be guaranteed by ensuring that the CBF remains nonnegative. In particular, it has been shown that, if the control input satisfies a particular linear constraint induced by the CBF for all time, then safety is ensured \cite{ames2016control}. CBFs have been generalized to high-degree   \cite{xiao2019control}, stochastic \cite{clark2019control,jagtap2020formal}, and uncertain \cite{xiao2020adaptive,cheng2019end} systems, and applied in diverse domains such as bipedal locomotion \cite{hsu2015control}, automotive control \cite{ames2014control,xiao2020bridging}, and UAVs \cite{xu2018safe,chen2020guaranteed}. 

While there is a large body of work on proving safety given CBF-induced linear control constraints, the problem of verifying whether such constraints can always be satisfied has received less attention. While constraint satisfaction is guaranteed for fully-actuated systems, the more general case remains open. Recent approaches have been proposed that consider finite time horizons \cite{jagtap2020formal,santoyo2021barrier} or assume that the control input is a ``backup controller'', however, such methods may be conservative. A systematic approach to verification of CBFs, as well as synthesizing CBFs with provable guarantees on safety, would enable broader adoption of CBFs as well as construction of CBFs with minimal performance trade-off.

In this paper, we develop a framework for verification and synthesis of CBFs for polynomial control systems. Our approach is based on the fact that the safety guarantees rely on two properties. First, there should be no point on the boundary of the CBF for which the Lie derivative of the CBF is always negative. Second, the safe region of the CBF should be contained within the overall safe region. We show that both conditions can be expressed as the non-existence of solutions to polynomial equations, which can then be mapped to a set of sum-of-squares inequalities via the Positivstellensatz \cite{stengle1974nullstellensatz}. This Positivstellensatz mapping allows verification of CBFs via semidefinite programming and construction of CBFs. Our specific contributions are as follows:
\begin{itemize}
\item We develop a SOS program for verifying that a given CBF construction ensures safety. We extend this approach to high-order CBFs and systems with multiple CBF constraints.
\item We propose an alternating-descent heuristic for synthesizing CBFs using the proposed conditions.
\item We consider special cases of our approach, including compact safe regions, and constraints where the unsafe region is a union of non-overlapping convex sets. 
\item We evaluate our approach through a numerical study.
\end{itemize}

This paper is organized as follows. Section \ref{sec:related} presents related work. Section \ref{sec:preliminaries} gives preliminary results. Section \ref{sec:verification} gives our problem formulation and CBF verification approach. Section \ref{sec:construction} presents our approach to CBF synthesis. Section \ref{sec:simulation} contains numerical results. Section \ref{sec:conclusion} concludes the paper.
\section{Related Work}
\label{sec:related}
CBFs were initially proposed in \cite{ames2014control} and further developed in \cite{hsu2015control}. Extensions of CBFs to high-degree \cite{nguyen2016exponential,xiao2019control}, stochastic \cite{clark2019control}, Euler-Lagrange \cite{barbosa2020provably,cortez2020correct}, uncertain \cite{folkestad2020data}, and hybrid \cite{jahanshahi2020synthesis} systems have been proposed. Minimal CBFs were proposed in \cite{konda2020characterizing}. CBFs with actuation constraints were investigated in \cite{chen2020guaranteed}. Most of these works provide safety guarantees under the assumption that a CBF satisfying given safety and reachability properties exists, or derive CBFs from a priori known backup controllers. Hence, they are complementary to the present paper, which investigates synthesis and verification of CBFs.

Some recent works have investigated the problem of constructing CBFs from the point of view of \emph{abstraction-free synthesis}, in which a higher-level specification (e.g., a temporal logic formula) is implemented by a sequence of control actions, each represented by one or more CBFs \cite{yang2019sampling,jagtap2020formal,santoyo2021barrier,anand2021compositional,niu2020control}. These approaches are distinct from this paper; in particular, they focus on safety over finite time horizons, as opposed to the positive set invariance guarantees provided in this paper. Sum-of-squares optimization techniques for safety verification, including SOS methods for barrier certificate construction, have been proposed \cite{prajna2007framework}. 

The problem of identifying positive invariant sets and characterizing states from which safety can be guaranteed has received extensive interest \cite{aubin2011viability,blanchini1999set}. Other related methods include HJB equations \cite{fisac2019bridging} and Lyapunov analysis \cite{romdlony2014uniting}. The results presented in this paper on verification of safety and synthesis of safe control policies, however, have not appeared in the existing literature. The problem of avoiding convex obstacles (considered in Section \ref{subsec:convex-unsafe} of this paper) has been studied, e.g., in \cite{khansari2012dynamical}, where it was shown that the obstacles can be avoided through a separating hyperplane construction. CBF-based policies for guaranteed convex obstacle avoidance, however, have not been proposed.
\section{Preliminaries and Background}
\label{sec:preliminaries}

In this section, we define sum-of-squares polynomials and provide background on the Positivstellensatz. A sum-of-squares (SOS) polynomial is a polynomial $f(x)$ such that 
\begin{displaymath}
f(x) = \sum_{i=1}^{m}{g_{i}(x)^{2}}
\end{displaymath}
for some polynomials $g_{1}(x),\ldots,g_{m}(x)$. Selecting coefficients of $f(x)$ to ensure that $f(x)$ is a sum-of-squares can be formulated as a semidefinite program, a procedure known as sum-of-squares optimization \cite{parrilo2003semidefinite}. 

The Positivstellensatz gives necessary and sufficient conditions for a set of polynomial constraints to have a real solution. It is stated as follows.
\begin{Theorem}[Positivstellensatz \cite{stengle1974nullstellensatz}]
\label{theorem:Positivstellensatz}
Consider a collection of polynomials $\{f_{j} : j=1,\ldots,r\}$, $\{g_{i} : i=1,\ldots,m\}$, and $\{h_{k} : k=1,\ldots,s\}$. Then the set 
\begin{multline*}
\left(\bigcap_{j=1}^{r}{\{x : f_{j}(x) \geq 0\}}\right) \cap \left(\bigcap_{i=1}^{m}{\{x : g_{i}(x) = 0\}}\right) \\
\cap \left(\bigcap_{k=1}^{s}{\{x : h_{k}(x) \neq 0\}}\right)
\end{multline*}
 is empty if and only if there exist polynomials $f$, $g$, and $h$ satisfying (i) $f(x) = \sum_{j=1}^{r}{\alpha_{j}(x)f_{j}(x)}$, where $\alpha_{j}(x)$ is a sum-of-squares; (ii) $g(x) = \sum_{i=1}^{m}{\eta_{i}(x)g_{i}(x)}$ for some polynomials $\eta_{1},\ldots,\eta_{m}$, (iii) $h$ is a product of powers of the $h_{k}$'s, and (iv) $f(x) + g(x) + h(x)^{2} = 0$ for all $x$.
\end{Theorem}

We note that we can verify that there does not exist a solution to $f_{j}(x) \geq 0$, $j=1,\ldots,r$ and $g_{i}(x) = 0$, $i=1,\ldots,m$ by choosing $\{h_{k}\}$ as the singleton set $h = 1$. The constraint then becomes $f(x) + g(x) + 1 = 0$ for all $x$.
\section{Problem Formulation: CBF Verification}
\label{sec:verification}
This section presents the problem of verifying a safe control policy. We first state the problem and then give our verification framework.
\subsection{Problem Statement and Background on CBFs}
\label{subsec:problem-statement}
We consider nonlinear control systems with dynamics
\begin{equation}
\label{eq:dynamics}
\dot{x}(t) = f(x(t)) + g(x(t))u(t)
\end{equation}
where $x(t) \in \mathbb{R}^{n}$ denotes the state, $u(t) \in \mathbb{R}^{m}$ is a control input, and  $f: \mathbb{R}^{n} \rightarrow \mathbb{R}^{n}$ and $g: \mathbb{R}^{n} \rightarrow \mathbb{R}^{m}$ are polynomials. The \emph{safe region} is defined as $\mathcal{C} = \{x: h(x) \geq 0\}$,  where $h: \mathbb{R}^{n} \rightarrow \mathbb{R}$ is a polynomial. The boundary of the safe region, denoted $\partial \mathcal{C}$, is defined by $\partial \mathcal{C} = \{x: h(x) = 0\}$. The \emph{viability kernel} is defined as follows \cite{gurriet2018online}.
\begin{definition}
\label{def:viability}
The viability kernel is a set $\Omega \subseteq \mathcal{C}$ such that, for any $x_{0} \in \Omega$, there is a control input signal $\{u(t) : t \in [0,\infty)\}$ that guarantees $x(t) \in \Omega$ for all time $t$ when the initial state is $x(0) = x_{0}$.
\end{definition}
The problem studied in this section is as follows.
\begin{problem}
\label{problem:verification}
Given a system (\ref{eq:dynamics}) and a safety constraint set $\mathcal{C}$, (i) verify whether a set $\tilde{\Omega}$ is contained in the viability kernel, and (ii) verify that a given control policy ensures that the system state remains in the viability kernel.
\end{problem}

Note that we do not consider any constraints on the control input $u(t)$ (e.g., limits on actuation). As a first step to our approach, we define control barrier functions and high-order control barrier functions as follows.

\begin{definition}
\label{def:CBF}
A function $b$ is a control barrier function for system (\ref{eq:dynamics}) if there is a class-K function $\alpha$ such that, for all $x$ with $b(x) \geq 0$, there exists $u$ satisfying 
\begin{equation}
\label{eq:CBF-def}
\frac{\partial b}{\partial x}(f(x) + g(x)u) \geq -\alpha(b(x))
\end{equation}
\end{definition}

The following result establishes the safety guarantees provided by CBFs.

\begin{Theorem}[\cite{ames2014control}]
\label{theorem:safety-CBF}
Suppose that $b$ is a CBF, $b(x(0)) \geq 0$, and $u(t)$ satisfies (\ref{eq:CBF-def}) for all $t$. Then the set $\{x: b(x) \geq 0\}$ is positive invariant.
\end{Theorem}

A potential drawback of CBFs occurs in high-degree systems when $\frac{\partial b}{\partial x}g(x) = 0$. High-order CBFs (HOCBFs), defined as follows, address this weakness.

\begin{definition}[\cite{xiao2019control}]
\label{def:HOCBF}
For a function $b$, define functions $\psi_{0},\ldots,\psi_{r}$ by 
\begin{IEEEeqnarray*}{rCl}
\psi_{0}(x) &=& b(x) \\
\psi_{1}(x) &=& \dot{\psi}_{0}(x) + \alpha_{1}(\psi_{0}(x)) \\
 & \vdots & \\
 \psi_{r}(x) &=& \dot{\psi}_{r-1}(x) + \alpha_{r}(\psi_{r-1}(x))
 \end{IEEEeqnarray*}
 where $\alpha_{1},\ldots,\alpha_{r}$ are class-K functions. Define set $C_{i}$ by $C_{i}= \{x : \psi_{i}(x) \geq 0\}$ for $i=0,\ldots,r$. The functions $b,\psi_{1},\ldots,\psi_{r}$ define a HOCBF if $\frac{\partial \psi_{i}}{\partial x}g(x) = 0$ for $i < r$ and, for all $x \in \bigcap_{i=0}^{r}{C_{i}}$, there exists $u$ satisfying
 \begin{equation}
 \label{eq:HOCBF-def}
 L_{f}^{r}b(x) + L_{g}L_{f}^{r-1}b(x)u + O(b(x)) + \alpha_{m}(\psi_{r-1}(x)) \geq 0
 \end{equation}
 where $O(b(x))$ denotes   lower-order Lie derivatives of $b(x)$.
 \end{definition}
 
 Throughout the paper, we let $L$ denote the Lie derivative of a function. The following result establishes the positive invariance of HOCBFs.
 \begin{Theorem}[\cite{xiao2019control}]
 \label{theorem:safety-HOCBF}
 If $b$ is a HOCBF, then the set $\bigcap_{i=0}^{r}{C_{i}}$ is positive invariant.
 \end{Theorem}
 
 While these results lead to safety guarantees for CBFs and HOCBFs, they do not provide a procedure for verifying that such an input $u(t)$ exists for all $t$ (Problem \ref{problem:verification}). We derive such an approach in the next section.
 
 \subsection{Verification of CBFs}
 \label{subsec:verification}
 In what follows, we present an SOS-based approach to verifying a CBF and approximating the viability kernel. We have the following initial result.
 \begin{proposition}
 \label{prop:CBF-verify}
 The function $b(x)$ is a CBF iff there is no $x$ satisfying $b(x) = 0$, $\frac{\partial b}{\partial x}g(x) = 0$, and $\frac{\partial b}{\partial x}f(x) < 0$.
 \end{proposition}
 
 \begin{proof}
 By Nagumo's Theorem \cite{aubin2011viability}, we have that $b(x(t)) \geq 0$ for all $t \geq 0$ iff $\dot{b}(x(t)) \geq0$ at the boundary, i.e., when $b(x) = 0$. Suppose that $b(x) = 0$. If $\frac{\partial b}{\partial x}g(x) \neq 0$, then $u(t)$ can be chosen to satisfy (\ref{eq:CBF-def}). One feasible choice is to select $i$ with $\left[\frac{\partial b}{\partial x}g(x)\right]_{i} \neq 0$ and choose $$u_{i}(t) = -\left(\left[\frac{\partial b}{\partial x}g(x)\right]_{i}\right)^{-1}(\frac{\partial b}{\partial x}f(x) + \alpha(b(x)))$$ and $u_{j}(t) = 0$ for $j \neq i$. Otherwise, if $\frac{\partial b}{\partial x}g(x) = 0$ and $\frac{\partial b}{\partial x}f(x) \geq 0$, then $\dot{b}(t) \geq 0$. Conversely, if $\frac{\partial b}{\partial x}g(x_{0}) = 0$, $\frac{\partial b}{\partial x}f(x_{0}) < 0$, and $b(x_{0}) = 0$ for some $x_{0}$, then the set $\{b(x) \geq 0\}$ will not be positive invariant, since $x(t)$ will leave the set if $x(0) = x_{0}$.
 \end{proof}
 
 Based on this proposition, we can formulate the following conditions via the Positivstellensatz.
 
 \begin{Theorem}
 \label{theorem:CBF-verify}
 A polynomial $b(x)$ is a CBF for (\ref{eq:dynamics}) if and only if there exist polynomials $\eta(x)$, $\theta_{1}(x),\ldots,\theta_{m}(x)$, SOS polynomials $\alpha_{0}(x)$ and $\alpha_{1}(x)$, and an integer $r$ such that
 \begin{multline}
 \label{eq:CBF-condition}
 \eta(x)b(x) + \sum_{i=1}^{m}{\theta_{i}(x)\left[\frac{\partial b}{\partial x}g(x)\right]_{i}}  \\
 - \alpha_{0}(x) - \alpha_{1}(x)\frac{\partial b}{\partial x}f(x) \\
 + \left(\frac{\partial b}{\partial x}f(x)\right)^{2r} = 0
 \end{multline}
 Moreover, the set $\{b(x) \geq 0\}$ is a subset of the viability kernel of $\mathcal{C}$  if (\ref{eq:CBF-condition}) holds and there exist SOS polynomials $\beta_{0}(x)$ and $\beta_{1}(x)$ and a polynomial $\omega(x)$ such that
 \begin{equation}
 \label{eq:CBF-feasible}
 \beta_{0}(x) + \beta_{1}(x)b(x) + w(x)h(x) + 1 = 0
 \end{equation}
 \end{Theorem}
 
 \begin{proof}
 By Proposition \ref{prop:CBF-verify}, we have that $b(x)$ is a CBF iff there is no $x$ satisfying $b(x) = 0$, $\frac{\partial b}{\partial x}g(x) = 0$, and $\frac{\partial b}{\partial x}f(x) < 0$. The latter condition is equivalent to $-\frac{\partial b}{\partial x}f(x) \geq 0$ and $\frac{\partial b}{\partial x}f(x) \neq 0$. These conditions are equivalent to (\ref{eq:CBF-condition}) by Theorem \ref{theorem:Positivstellensatz}.
 
 Now, suppose that (\ref{eq:CBF-condition}) holds, and hence $b(x(t)) \geq 0$ for all $t$ under any control policy satisfying (\ref{eq:CBF-def}) for all $t$. It therefore suffices to show that $b(x(t)) \geq 0$ for all $t$ implies that $h(x(t)) \geq 0$ for all $t$. If this is not the case, then there exists a point satisfying $h(x) =0$ and $b(x) \geq 0$. This, however, contradicts (\ref{eq:CBF-feasible}) by Theorem \ref{theorem:Positivstellensatz}.
 \end{proof}
 
 Next, we derive equivalent conditions for $\psi_{0}(x),\ldots,\psi_{r}(x)$ to define a HOCBF. 
 \begin{Theorem}
 \label{theorem:HOCBF-verify}
 If $b(x)$ is a polynomial with relative degree $r$ with respect to (\ref{eq:dynamics}), then $b(x)$ and $\psi_{0},\ldots,\psi_{r}$ define a HOCBF if and only if there exist polynomials $\eta(x), \theta_{1}(x),\ldots,\theta_{r}(x)$, SOS polynomials $\alpha_{0}(x), \alpha_{1}(x)$, and $\gamma_{0}(x),\ldots,\gamma_{r-1}(x)$, and an integer $s$ such that 
 \begin{multline}
 \label{eq:HOCBF-condition}
 \eta(x)b(x) + \sum_{i=1}^{r}{\theta_{i}(x)L_{g}L_{f}^{r-1}b(x)} \\
 -\alpha_{0}(x)-\alpha_{1}(x)(L_{f}^{r}b(x) + O(b(x))) + \left(L_{f}^{r}b(x) + O(b(x))\right)^{2s} \\
 + \sum_{i=0}^{r-1}{\gamma_{i}(x)\psi_{i}(x)} = 0
 \end{multline}
 Moreover, the set $\bigcap_{i=0}^{r}{C_{i}}$ is a subset of the viability kernel if there exist SOS polynomials $\beta_{0}(x),\ldots,\beta_{r+1}(x)$ and a polynomial $w(x)$ such that 
\begin{equation}
\label{eq:HOCBF-feasible} 
 \sum_{i=0}^{r}{\psi_{i}(x)\beta_{i}(x)} + \beta_{r+1}(x) + w(x)h(x) + 1 = 0.
 \end{equation}
 \end{Theorem}
 \begin{proof}
 By Theorem \ref{theorem:safety-HOCBF}, positive invariance of $\{b(x) \geq 0\}$ is guaranteed if (\ref{eq:HOCBF-condition}) holds for some $u$ at each $x \in \bigcap_{i=0}^{r}{C_{i}}$. Such a $u$ exists if there is no $x$ with $x \in C_{i}$ for all $i=0,\ldots,r$, $\psi_{r}(x) = 0$, $L_{f}^{r-1}L_{g}b = 0$, and $L_{f}^{r}b + O(b(x)) < 0$. This is equivalent to (\ref{eq:HOCBF-condition}) by Theorem \ref{theorem:Positivstellensatz}.
 
 Now, we want to show that if $x(t) \in C_{i}$ for $i=0,\ldots,r$ and $h(x(0)) \geq 0$, then $h(x(t)) \geq 0$ for all $t$ if (\ref{eq:HOCBF-feasible}) holds. If (\ref{eq:HOCBF-feasible}) holds, then there is no $x$ satisfying $x \in C_{i}$ for $i=0,\ldots,r$ and $h(x) = 0$. Since $x(t) \in C_{i}$ for all $t$ and $i=0,\ldots,r$, $x(t)$ never reaches the boundary of $\mathcal{C}$, and hence remains in $\mathcal{C}$ for all time $t$.  
 \end{proof}
 
 We next consider intersections  of CBFs. Such constructions may be useful when a single CBF is not sufficient to characterize the viability kernel.  An example where such constraints are relevant is when the safe region consists of the complement of a set of disjoint compact sets, e.g., in an obstacle avoidance problem where each of the sets represents an obstacle.
 \begin{Theorem}
 \label{theorem:intersection-CBF}
 Let $b_{1}(x),\ldots,b_{k}(x)$ be a collection of polynomials. There is a control policy that renders the set $\bigcap_{i=1}^{k}{\{x : b_{i}(x) \geq 0\}}$ positive invariant if the following conditions hold:
 \begin{enumerate}
 \item For each $i=1,\ldots,k$, there exist polynomials $\eta_{i}(x)$, $\theta_{i,1}(x),\ldots,\theta_{i,m}(x)$, and SOS polynomials $\alpha_{i,0}(x)$, $\alpha_{i,1}(x)$, $\{\lambda_{ij}(x): j \in \{1,\ldots,k\} \setminus \{i\}\}$ such that 
 \begin{multline}
 \label{eq:multi-CBF-condition-1}
 \eta_{i}(x)b_{i}(x) + \sum_{j=1}^{m}{\theta_{i,j}(x)\left[\frac{\partial b_{i}}{\partial x}g(x)\right]_{j}} + \alpha_{i,0}(x) \\ -\alpha_{i,1}(x)\frac{\partial b_{i}}{\partial x}f(x) 
 + \sum_{j \neq i}{\lambda_{ij}(x)b_{j}(x)} + 1 = 0
 \end{multline}
 \item For each $i,j \in \{1,\ldots,k\}$ with $i \neq j$, there exist polynomials $w_{ij}^{i}(x)$ and $w_{ij}^{j}(x)$ and SOS polynomials $\alpha_{i,j,l}(x)$ for $l=1,\ldots,m$, such that 
 \begin{multline}
 \label{eq:multi-CBF-condition-2}
 -\sum_{l=1}^{m}{\left(\alpha_{i,j,l}(x)\left[\frac{\partial b_{i}}{\partial x}g(x)\right]_{l}\left[\frac{\partial b_{j}}{\partial x}g(x)\right]_{l}\right)} \\
 + \alpha_{i,j,0}(x) + w_{ij}^{i}(x)h_{i}(x) + w_{ij}^{j}(x)h_{j}(x) + 1 = 0
 \end{multline}
 \end{enumerate}
 \end{Theorem}
 
 \begin{proof}
 We consider a control policy in which the control $u$ is chosen to satisfy (\ref{eq:CBF-def}) for each of the functions $b_{i}$. Suppose the conditions of the theorem hold. For each $x(t)$ at the boundary of $\bigcap_{i=1}^{k}{\{b_{i}(x) \geq 0\}}$, define $Z(x) = \{i : b_{i}(x) = 0\}$. We have the following cases.
 
 \emph{\underline{Case I -- $Z(x) = \{i\}$:}} In this case, Condition (\ref{eq:multi-CBF-condition-1}) implies that there is no $x \in \bigcap_{j=1}^{k}{\{b_{j}(x) \geq 0\}}$ with $b_{i}(x) = 0$, $\frac{\partial b_{i}}{\partial x}g(x) = 0$, and $\frac{\partial b_{i}}{\partial x}f(x) < 0$. Hence we can select a $u(t)$ satisfying (\ref{eq:CBF-def}), ensuring that the state does not leave $\bigcap_{i=1}^{m}{\{b_{i}(x) \geq 0\}}$.
 
 \emph{\underline{Case II - $|Z(x)| > 1$:}} We claim that, if Condition (\ref{eq:multi-CBF-condition-2}) holds, then for all $i,j \in Z(x)$ and $l=1,\ldots,m$, either $\left[\frac{\partial b_{i}}{\partial x}g(x)\right]_{l}$ and $\left[\frac{\partial b_{j}}{\partial x}g(x)\right]_{l}$ have the same sign, or one of them is zero. Indeed, Condition (\ref{eq:multi-CBF-condition-2}) and the Positivstellensatz imply that there is no point $x$ satisfying $h_{i}(x) = h_{j}(x) = 0$ and $$\left[\frac{\partial b_{i}}{\partial x}g(x)\right]_{l}\left[\frac{\partial b_{j}}{\partial x}g(x)\right]_{l}, \ l=1,\ldots,m.$$ Hence, for each $i \in Z(x)$, we can select $u$ of sufficient magnitude to ensure that constraint (\ref{eq:CBF-def}) is satisfied.
 \end{proof}
 
 Finally, we note that, given multiple CBFs $b_{1}(x), \ldots, b_{k}(x)$ for a safe region $\mathcal{C}$, the set $\bigcup_{i=1}^{k}{\{x : b_{i}(x) \geq 0\}}$ is contained in the viability kernel.
 
\section{Control Barrier Function Construction}
\label{sec:construction}
This section gives techniques for constructing control barrier functions. The problem studied in this section is formulated as follows.
\begin{problem}
\label{problem:synthesis}
Given a system (\ref{eq:dynamics}) and a safety constraint set $\mathcal{C}$, construct a CBF (or HOCBF) $b$ that verifiably guarantees positive invariance of $\mathcal{C}$.
\end{problem}

 We first present a general heuristic based on sum-of-squares optimization, followed by approaches for linear systems and compact safe regions, as well as convex unsafe regions.

\subsection{General Heuristic}
\label{subsec:general-heuristic}
Our approach is based on the observation that the CBF conditions of Theorems \ref{theorem:CBF-verify}--\ref{theorem:intersection-CBF} are polynomial constraints, enabling efficient verification when the candidate CBF $b(x)$ is given. If $b(x)$ is unknown, however, the constraints are nonlinear and cannot be converted to linear matrix inequalities as in sum-of-squares optimization. We therefore propose an alternating-descent heuristic which is described as follows.

We initialize $b^{0}(x)$ arbitrarily, and initialize parameters $\rho_{0}$ and $\rho^{\prime}$ to be infinite. At step $k$, we solve the following SOS program with variables $\rho$, $\alpha^{k}(x)$, $\eta^{k}(x)$, $\theta_{1}^{k}(x),\ldots,\theta_{m}^{k}(x)$, $w(x)$, and $\beta^{k}(x)$
\begin{equation}
\label{eq:synthesis-SOS-1}
\begin{array}{ll}
\mbox{minimize} & \rho \\
\mbox{s.t.} & \alpha^{k}(x)\frac{\partial b^{k-1}}{\partial x}f(x) + \sum_{i=1}^{m}{\theta_{i}^{k}(x)\left[\frac{\partial b^{k-1}}{\partial x}g(x)\right]_{i}} \\
 & + \eta^{k}(x)b^{k-1}(x) + \rho\Lambda(x) -1 \in SOS \\
  & -\beta^{k}(x)b^{k-1}(x) + w(x)h(x) -1 \in SOS \\
  & \beta^{k}(x), \alpha^{k}(x) \in SOS
  \end{array}
  \end{equation}
  
 Here $\Lambda(x)$ is a fixed SOS polynomial of sufficiently large degree to ensure that the first constraint is SOS for sufficiently large $\rho$. We let $\rho_{k}$ denote the value of $\rho$ returned by the optimization. The procedure terminates if $\rho \leq 0$ or if $|\rho_{k}-\rho_{k-1}^{\prime}| < \epsilon$. Otherwise, we solve the SOS problem with variables $\rho$ and $b^{k}$ given by 
 \begin{equation}
 \label{eq:synthesis-SOS-2}
 \begin{array}{ll}
\mbox{minimize} & \rho \\
\mbox{s.t.} & \alpha^{k}(x)\frac{\partial b^{k}}{\partial x}f(x) + \sum_{i=1}^{m}{\theta_{i}^{k}(x)\left[\frac{\partial b^{k}}{\partial x}g(x)\right]_{i}} \\
 & + \eta^{k}(x)b^{k}(x) + \rho\Lambda(x) - 1 \in SOS \\
  & -\beta^{k}(x)b^{k}(x) + w(x)h(x) - 1\in SOS \\
  & \beta^{k}(x), \alpha^{k}(x) \in SOS
  \end{array}
  \end{equation}
  We let $\rho_{k}^{\prime}$ denote the value of $\rho$ returned by the optimization. The procedure terminates if $\rho \leq 0$ or if $|\rho_{k}-\rho_{k}^{\prime}| < \epsilon$. The following theorem establishes the correctness and convergence of this approach.
  \begin{Theorem}
  \label{theorem:CBF-heuristic-correct}
  If the procedure described above terminates at stage $k$ with $\rho_{k} \leq 0$, then $b^{k-1}(x)$ is a CBF and $\{b^{k-1}(x) \geq 0\}$ is a subset of the viability kernel. If the procedure terminates at stage $k$ with $\rho_{k}^{\prime} \leq 0$, then $b^{k}(x)$ is a CBF and $\{b^{k}(x) \geq 0\}$ is a subset of the viability kernel. The procedure converges in a finite number of iterations $k$.
  \end{Theorem}
  
  \begin{proof}
  We prove that if $\rho_{k} \leq 0$, then $b^{k-1}(x)$ is a CBF and $\{b^{k-1}(x) \geq 0\}$ is a subset of the viability kernel. The case of $\rho_{k}^{\prime} \leq 0$ is similar. If $\rho_{k} \leq 0$, then by construction of (\ref{eq:synthesis-SOS-1}) there exist SOS polynomials $\overline{\alpha}(x)$ and $\overline{\beta}(x)$ such that
  \begin{IEEEeqnarray}{rCl}
  \overline{\alpha}(x) &=& \alpha^{k}(x)\frac{\partial b^{k-1}}{\partial x}f(x) + \sum_{i=1}^{m}{\theta_{i}^{k}(x)\left[\frac{\partial b^{k-1}}{\partial x}g(x)\right]_{i}} \\
\label{eq:CBF-heuristic-1}
  && + \eta^{k}(x)b^{k-1}(x) + \rho_{i}\Lambda(x) - 1 \\
  \label{eq:CBF-heuristic-2}
  \overline{\beta}(x) &=& -\beta^{k}(x)b^{k-1}(x) + w(x)h(x) - 1
  \end{IEEEeqnarray}
  Rearranging (\ref{eq:CBF-heuristic-1}) and setting $\overline{\theta} = -\theta_{i}^{k}$ and $\overline{\eta} = -\eta^{k}$ yields
  \begin{multline*}
  \sum_{i=1}^{m}{\overline{\theta}_{i}(x)\left[\frac{\partial b^{k-1}}{\partial x}g(x)\right]_{i}} + \overline{\eta}(x)b^{k-1}(x) \\
  + \overline{\alpha}(x) - \rho_{k}\Lambda(x) - \alpha^{k}(x)\frac{\partial b^{k-1}}{\partial x}f(x) + 1 = 0
  \end{multline*}
  Since $\rho_{k} \leq 0$ and $\Lambda(x)$ is SOS, $\overline{\alpha}(x) - \rho_{k}\Lambda(x)$ is SOS. Hence, by Theorem \ref{theorem:Positivstellensatz}, there does not exist any $x$ satisfying $$\frac{\partial b^{k-1}}{\partial x}f(x) \leq 0, \ b^{k-1}(x)=0, \ \frac{\partial b^{k-1}}{\partial x}g(x) = 0$$ and $b^{k-1}(x)$ is a CBF by Proposition \ref{prop:CBF-verify}.
  
  Similarly, setting $\overline{w}(x) = -w(x)$, we have $$\overline{\beta}(x) + \beta^{k}(x)b^{k-1}(x) + \overline{w}(x)h(x) + 1= 0,$$ implying that $\{b^{k-1}(x) \geq 0\} \subseteq \mathcal{C}$ by Theorem \ref{theorem:CBF-verify}.
  
  Finally, to prove convergence, we observe that $\rho_{k-1}^{\prime}$, $\alpha^{k-1}(x)$, $\eta^{k-1}(x)$, $\theta_{1}^{k-1}(x),\ldots,\theta_{m}^{k-1}(x)$, and $\beta^{k-1}(x)$ comprise a feasible solution to (\ref{eq:synthesis-SOS-1}) and $\rho_{k}$, $b^{k-1}(x)$ are feasible solutions to (\ref{eq:synthesis-SOS-2}). Hence $\rho_{1} \geq \rho_{1}^{\prime} \geq \cdots \geq \rho_{k} \geq \rho_{k}^{\prime} \geq \cdots$, i.e., the sequence is monotone nonincreasing. If the sequence is unbounded, then there exists $K$ such that $\rho_{k},\rho_{k}^{\prime} \leq 0$ for all $k > K$, and hence the algorithm terminates after $K$ iterations. If the sequence is bounded below by zero, then it converges by the monotone convergence theorem and hence $|\rho_{k}-\rho_{k}^{\prime}| < \epsilon$ and $|\rho_{k}-\rho_{k-1}^{\prime}| < \epsilon$ for $k$ sufficiently large, implying termination of the procedure.
  \end{proof}
  
  We next consider the problem of constructing HOCBFs. We observe that, in addition to the same non-convexity as the CBF, HOCBFs have an additional degree of freedom, namely, the selection of the class-K functions $\alpha_{i}$ in Definition \ref{def:HOCBF}. To mitigate this complexity, we present the following.
  
  \begin{proposition}
  \label{prop:HOCBF-alternate}
  Let $\psi_{0}(x),\ldots,\psi_{r}(x)$ be a collection of polynomials such that the following hold: (i) for all $i=1,\ldots,r$, there is no $x$ satisfying $\psi_{i}(x) \geq 0$, $\psi_{i-1}(x) = 0$, 
  and $\frac{\partial \psi_{i-1}}{\partial x}f(x) < 0$, and (ii) $\frac{\partial \psi_{i}}{\partial x}g(x) = 0$ for all $i < r$. If (\ref{eq:HOCBF-def}) holds for all $t$, then $\bigcap_{i=0}^{r}{\{\psi_{i}(x) \geq 0\}}$ is positive invariant.
  \end{proposition}
  
  \begin{proof}
  The proof is by backwards induction. By (\ref{eq:HOCBF-def}), $\psi_{r}(x(t)) \geq 0$ for all $t$. By Nagumo's Theorem, we have that $L\psi_{i-1}(x) \geq 0$ at the boundary points where $\psi_{i-1}(x)=0$ and $\psi_{i} \geq 0$, and hence $\{\psi_{i-1} \geq 0\}$ is positive invariant if $\{\psi_{i}(x) \geq 0\}$ is positive invariant, completing the proof.
  \end{proof}
  
  We then adopt the following alternating-descent procedure analogous to the CBF case. We initialize $\psi_{0}$ arbitrarily and choose $\psi_{1},\ldots,\psi_{r}$ according to Definition \ref{def:HOCBF}. We then solve the SOS problem with variables $\rho$, $\alpha^{k}(x)$, $\eta^{k}(x)$, $\theta_{1}^{k}(x),\ldots,\theta_{m}^{k}(x)$, $w(x)$, $\beta^{k}(x)$, $\lambda_{0}^{k}(x),\ldots,\lambda_{r-1}^{k}(x)$, $\zeta_{0}^{k}(x),\ldots,\zeta_{r-1}^{k}(x)$, $\gamma_{0}^{k}(x),\ldots,\gamma_{r-1}^{k}(x)$, and $\phi_{0}(x),\ldots,\phi_{r-1}(x)$
  \begin{equation}
  \label{eq:HOCBF-synthesis-SOS-1}
  \begin{array}{ll}
  \mbox{min} & \rho_{k} \\
  \mbox{s.t.} & \eta^{k}(x)\psi_{r}^{k-1}(x) + \sum_{i=1}^{m}{\theta_{i}^{k}(x)L_{g}L_{f}^{r-1}\psi_{r-1}^{k-1}(x)} \\
   & + \alpha^{k}(x)(L_{f}^{r}\psi_{r-1}^{k-1}(x) + O(b(x))) \\
   & - \sum_{i=0}^{r-1}{\gamma_{i}^{k}(x)\psi_{i}^{k-1}(x)} + \rho_{k}\Lambda(x) -1 \in SOS \\
   & -\sum_{i=0}^{r-1}{\psi_{i}^{k-1}(x)\beta_{i}^{k}(x)} + w(x)h(x) -1 \in SOS \\
   & -\lambda_{i}^{k}\psi_{i}^{k-1}(x) + \zeta_{i-1}^{k}(x)\psi_{i-1}^{k-1} \\
   &+ \phi_{i}^{k}(x)L\psi_{i-1}^{k-1}(x) -1\in SOS, \quad i=1,\ldots,r \\
   & \lambda_{i}^{k}(x),\phi_{i}^{k}(x),\beta_{i}^{k}(x),\alpha^{k}(x), \gamma_{i}^{k}(x) \in SOS
   \end{array}
   \end{equation}
  As above, $\Lambda(x)$ is a suitably-chosen SOS polynomial. If the solution $\rho_{k} \leq 0$ or $|\rho_{k}-\rho_{k-1}^{\prime}| < \epsilon$, the algorithm terminates. Else, we  solve the SOS program with variables $\psi_{0}^{k}(x),\ldots,\psi_{r}^{k}(x), \rho_{k}^{\prime}$ given by 
    \begin{equation}
  \label{eq:HOCBF-synthesis-SOS-2}
  \begin{array}{ll}
  \mbox{min} & \rho_{k}^{\prime} \\
  \mbox{s.t.} & \eta^{k}(x)\psi_{r-1}^{k}(x) + \sum_{i=1}^{m}{\theta_{i}^{k}(x)L_{g}L_{f}^{r-1}\psi_{r-1}^{k}(x)} \\
   & + \alpha^{k}(x)(L_{f}^{r}\psi_{r-1}^{k}(x) + O(b(x))) - \sum_{i=0}^{r-1}{\gamma_{i}^{k}(x)\psi_{i}^{k}(x)} \\
   & + \rho_{k}^{\prime}\Lambda(x) - 1 \in SOS \\
   & -\sum_{i=0}^{r}{\psi_{i}^{k}(x)\beta_{i}^{k}(x)} + w(x)h(x) - 1 \in SOS \\
   & -\lambda_{i}^{k}\psi_{i}^{k}(x) + \zeta_{i-1}^{k}(x)\psi_{i-1}^{k} \\
   & \frac{\partial \psi_{i}^{k}}{\partial x}g(x) = 0, \ i=0,\ldots,(r-1) \\
   & + \phi_{i}^{k}(x)L\psi_{i-1}^{k}(x) - 1 \in SOS,  \qquad i=1,\ldots,r \\
   \end{array}
   \end{equation}
   If $\rho_{k}^{\prime} \leq 0$ or $|\rho_{k}-\rho_{k}^{\prime}| < \epsilon$, the algorithm terminates. We have the following result.
   \begin{Theorem}
   \label{theorem:HOCBF-heuristic-correct}
   If $\rho_{k} \leq 0$, then $\psi_{0}^{k-1}(x),\ldots,\psi_{r}^{k-1}(x)$ define a HOCBF and $\bigcap_{i=0}^{r}{\{\psi_{i}^{k-1}(x)\}} \subseteq \mathcal{C}$. If $\rho_{k}^{\prime} \leq 0$, then $\psi_{0}^{k},\ldots,\psi_{r}^{k}(x)$ define HOCBF and  $\bigcap_{i=0}^{r}{\{\psi_{i}^{k}(x)\}} \subseteq \mathcal{C}$. The procedure terminates in finite time.
   \end{Theorem}
   
   \begin{proof}
   We prove that if $\rho_{k} \leq 0$, then $\psi_{0}^{k-1}(x),\ldots,\psi_{r}^{k-1}(x)$ defines a HOCBF and $\bigcap_{i=0}^{r}{\{\psi_{i}^{k-1}(x)\}} \subseteq \mathcal{C}$. The case of $\rho_{k}^{\prime} \leq 0$ is similar. If $\rho_{k} \leq 0$, then
   \begin{multline*}
   \overline{\eta}(x)\psi_{r}^{k-1}(x) - \alpha^{k}(x)(L_{f}^{r}\psi_{r}^{k-1}(x) + O(b(x))) + \overline{\alpha}(x)  \\
    + \sum_{i=1}^{m}{\overline{\theta}_{i}(x)L_{g}L_{f}^{r-1}\psi_{r}^{k-1}(x)} + \sum_{i=0}^{r-}{\gamma_{i}^{k}(x)\psi_{i}^{k-1}(x)}  + 1  = 0
    \end{multline*}
    for some SOS polynomial $\overline{\alpha}(x)$, where $\overline{\eta}(x) = -\eta^{k}(x)$ and $\overline{\theta}_{i}(x) = -\theta_{i}^{k}(x)$. Hence Eq. (\ref{eq:HOCBF-def}) is feasible for all $x \in \bigcap_{i=0}^{r}{\{\psi_{i}^{k-1}(x) \geq 0\}}$ by the Positivstellensatz. Similarly, we have that $$\overline{\beta}(x) + \sum_{i=0}^{r}{\psi_{i}^{k-1}(x)\beta_{i}^{k}(x)} + w(x)h(x),$$ and hence $\bigcap_{i=0}^{r}{\{\psi_{i}^{k-1}(x) \geq 0\}}$ is in the viability kernel. Finally, we have 
\begin{multline*}    
    \lambda_{i}^{k}\psi_{i}^{k-1}(x)  - \zeta_{i-1}(x)\psi_{i}^{k-1}(x) -\phi_{i}^{k}(x)L\psi_{i-1}^{k-1}(x) \\
    + \overline{\lambda}_{i}(x) + 1 = 0
\end{multline*}    
     for some SOS polynomial $\overline{\lambda}_{i}(x)$ for all $i$, and thus the conditions of Proposition \ref{prop:HOCBF-alternate} hold and the set is a subset of the viability kernel. The termination proof is similar to  Theorem \ref{theorem:CBF-heuristic-correct}.
   \end{proof}
   
   The preceding results provide heuristics for constructing CBFs and HOCBFs. However, we observe several limitations of this approach. First, there is no guarantee that the algorithm will terminate when $\rho_{k} \leq 0$, and hence no guarantee that a CBF or HOCBF will be returned. Second, the approach involves solving $2k$ SOS programs, each of which includes multiple linear matrix inequalities, and hence is computationally challenging for high-dimensional systems. In the following, we present more efficient approaches with provable guarantees for special cases of compact safe regions, as well as convex unsafe regions.
   
   \subsection{Compact Safe Region}
   \label{subsec:compact}
   In what follows, we consider the case where $\mathcal{C}$ is compact. Our approach is motivated by the following preliminary result.
   \begin{proposition}[\cite{richeson2002fixed}]
   \label{prop:compact-fixed-point}
   Suppose that $\dot{x}(t) = f(x(t))$ is a dynamical system that remains in a compact region. Then $f(x)$ has a fixed point.
   \end{proposition}
  
   We observe that, if a given feedback control policy ensures that $x(t)$ remains within the safe region, then the control policy results in a fixed point. Based on this preliminary result, our approach is to identify fixed points of (\ref{eq:dynamics}), and then construct barrier functions in the neighborhoods of those points. 
   \begin{Lemma}
   \label{lemma:fixed-point}
   Let $(x^{\ast},u^{\ast})$ satisfy $f(x^{\ast}) + g(x^{\ast})u^{\ast} = 0$ and $x^{\ast} \in \mbox{int}(\mathcal{C})$. Then there exists a positive definite matrix $P$ and $\delta > 0$ such that the function $$b_{0}(x; x^{\ast}) = \delta-(x-x^{\ast})^{T}P(x-x^{\ast})$$ is a CBF with $\{b_{0}(x ; x^{\ast}) \geq 0\} \subseteq \mathcal{C}$.
   \end{Lemma}
   
   \begin{proof}
    Let $\dot{x}(t) = Fx(t) + Gu(t)$ be the linearization of (\ref{eq:dynamics}) in a neighborhood of $(x^{\ast},u^{\ast})$. Let $\tilde{u}(t) = {u}^{\ast} - K(x-x^{\ast})$ be a stabilizing controller of the linearized system. Hence, there is a quadratic Lyapunov function $V(x) = (x-x^{\ast})^{T}P(x-x^{\ast})$ such that  (\ref{eq:dynamics}) is asymptotically stable in a neighborhood of $x^{\ast}$, i.e., $$\frac{\partial V}{\partial x}(f(x) + g(x)\tilde{u}(t)) \leq 0.$$ By construction of $b_{0}(x ; x^{\ast})$, this is equivalent to $$\frac{\partial b_{0}(x ; x^{\ast})}{\partial x}(f(x) + g(x)\tilde{u}) \geq 0$$ for $(x-x^{\ast})$ sufficiently small. Furthermore, since $x^{\ast}$ is an interior point, $b_{0}(x ; x^{\ast}) \geq 0$ is contained in $\mathcal{C}$ for $\delta$ sufficiently small.
    \end{proof}
    
    The preceding lemma implies that we can construct a CBF by identifying a fixed point $(x^{\ast},u^{\ast})$, constructing a stabilizing linear controller $\overline{u} = -K\overline{x}$ of the linearized system, and solving the Lyapunov equation $$\overline{F}^{T}P + P\overline{F} + N = 0,$$ where $\overline{F} = (F-GK)$ and $N$ is a positive definite matrix. We can then to find a value of $\delta$ such that $b_{0}(x ; x^{\ast})$ is a CBF by solving a sequence of SOS verification problems of the form (\ref{eq:CBF-condition}) and (\ref{eq:CBF-feasible}).    We have the following result.

    
    \begin{Theorem}
    \label{theorem:compact-CBF}
    Suppose that $(x^{\ast},u^{\ast})$ is a fixed point of (\ref{eq:dynamics}), $P$ is a solution to the Lyapunov equation corresponding to a stabilizing controller of the linearized system, and the SOS constraints 
    \begin{multline}
    \label{eq:compact-SOS-1}
    \alpha(x)(-2(x-x^{\ast})^{T}Pf(x)) + \sum_{i=1}^{m}{\theta_{i}(x)\left[-2(x-x^{\ast})^{T}Pg(x)\right]_{i}} \\
    + \eta(x)(\delta-(x-x^{\ast})^{T}P(x-x^{\ast})) \in SOS
    \end{multline}
    and
    \begin{multline}
    \label{eq:compact-SOS-2}
    h(x)-\beta(x)(\delta - (x-x^{\ast})^{T}P(x-x^{\ast})) \in SOS
    \end{multline}
    hold for some polynomials $\theta_{1},\ldots,\theta_{m}, \eta(x)$ and SOS polynomials $\alpha(x)$ and $\beta(x)$. 
    Then $b_{0}(x ; x^{\ast}) = \delta - (x-x^{\ast})^{T}P(x-x^{\ast})$ is a CBF and $\{b_{0}(x ; x^{\ast}) \geq 0\}$ is in the viability kernel.
    \end{Theorem}
    
    \begin{proof}
    If (\ref{eq:compact-SOS-1}) holds, then $b_{0}$ is a CBF by Theorem \ref{theorem:CBF-verify}. To show (\ref{eq:compact-SOS-2}) implies that $\{b_{0} \geq 0\}$ is in the viability kernel, we use nonnegativity of SOS polynomials to conclude that $\{b_{0}(x) \geq 0\}$ implies $\{h(x) \geq 0\}$ if $$h(x) = \beta(x)b_{0}(x) + \beta_{0}(x)$$ for some SOS $\beta$ and $\beta_{0}$. Rearranging terms gives (\ref{eq:compact-SOS-2}).
    \end{proof}
    
    The preceding theorem suggests that a binary search approach can be used to find the maximum $\delta$ such that (\ref{eq:compact-SOS-1}) and (\ref{eq:compact-SOS-2}) hold, and choosing $b_{0}(x ; x^{\ast})$ as the CBF. Moreover, we can generalize this approach by considering functions of the form $$b_{1}(x ; x^{\ast}) = b_{0}(x; x^{\ast}) + \alpha(h(x)),$$ where $\alpha(\cdot)$ is a class-K function. We have the following.
    \begin{Lemma}
    \label{lemma:CBF-compact-construction}
    If $b_{1}(x ; x^{\ast})$ is a CBF, then the set $$\mathcal{V}(b_{0},b_{1}) = \left(\{b_{0}(x ; x^{\ast}) \geq 0\} \cup \{b_{1}(x ; x^{\ast}) \geq 0\}\right) \cap \mathcal{C}$$ is in the viability kernel.
    \end{Lemma}
    
    \begin{proof}
    We construct a control policy that ensures positive invariance of $\mathcal{V}(b_{0},b_{1})$ as follows. We select $\mu$ to satisfy (\ref{eq:CBF-def}) for $b_{0}(x ; x^{\ast})$ when $b_{0}(x(t) ; x^{\ast}) \geq 0$ and to satisfy (\ref{eq:CBF-def}) for $b_{1}(x ; x^{\ast}) \geq 0$ if $b_{1}(x(t) ; x^{\ast}) \geq 0$ and $b_{0}(x ; x^{\ast}) < 0$.
    
    First, suppose that $b_{0}(x(t^{\prime}) ; x^{\ast}) \geq 0$ for some $t^{\prime}$. Since $b_{0}(x ; x^{\ast})$ is a CBF, the set $\{b_{0}(x ; x^{\ast}) \geq 0\}$ is positive invariant and is contained in $\mathcal{C}$. 
    
    Next, suppose that $x(t_{0}) \in \mathcal{V}(b_{0},b_{1}) \setminus \{b_{0}(x) \geq 0\}$ and yet $x(t_{0})$ is not in the viability kernel. Let $$t^{\ast} = \inf{\{t : h(x(t)) = 0\}}.$$ By the preceding discussion, we must have $b_{0}(x(t)) < 0$ for $t \in [t_{0},t^{\ast}]$. We therefore have
    \begin{IEEEeqnarray*}{rCl}
    b_{1}(x(t^{\ast}) ; x^{\ast}) &=& b_{0}(x(t^{\ast}) ; x^{\ast}) + \alpha(h(x(t^{\ast})) \\
    &=& b_{0}(x(t^{\ast})) < 0
    \end{IEEEeqnarray*}
    contradicting Theorem \ref{theorem:safety-CBF}.
    \end{proof}
    One candidate for $\alpha$ is the function $\alpha(h(x)) = kh(x)$, where $k$ is a positive real number. We observe that, for $k$ sufficiently small, $b_{0}(x ; x^{\ast}) + kh(x)$ is a CBF. We can therefore select such a $k$ by solving the feasibility problem (\ref{eq:CBF-condition})--(\ref{eq:CBF-feasible}) and approximate the viability kernel as $$\bigcup_{x^{\ast} \in FP}{\left(\left(\{b_{0}(x ; x^{\ast}) \geq 0\} \cup \{b_{0}(x; x^{\ast}) + kh(x) \geq 0\}\right) \cap \mathcal{C}\right)}$$ where FP denotes the set of fixed points of (\ref{eq:dynamics}).
    
    We observe that, while this approach is valid for general unsafe regions, there may be safe trajectories that do not contain fixed points when the safe region is not compact.
    
    \begin{figure*}
\centering
$\begin{array}{cc}
\includegraphics[width=3in]{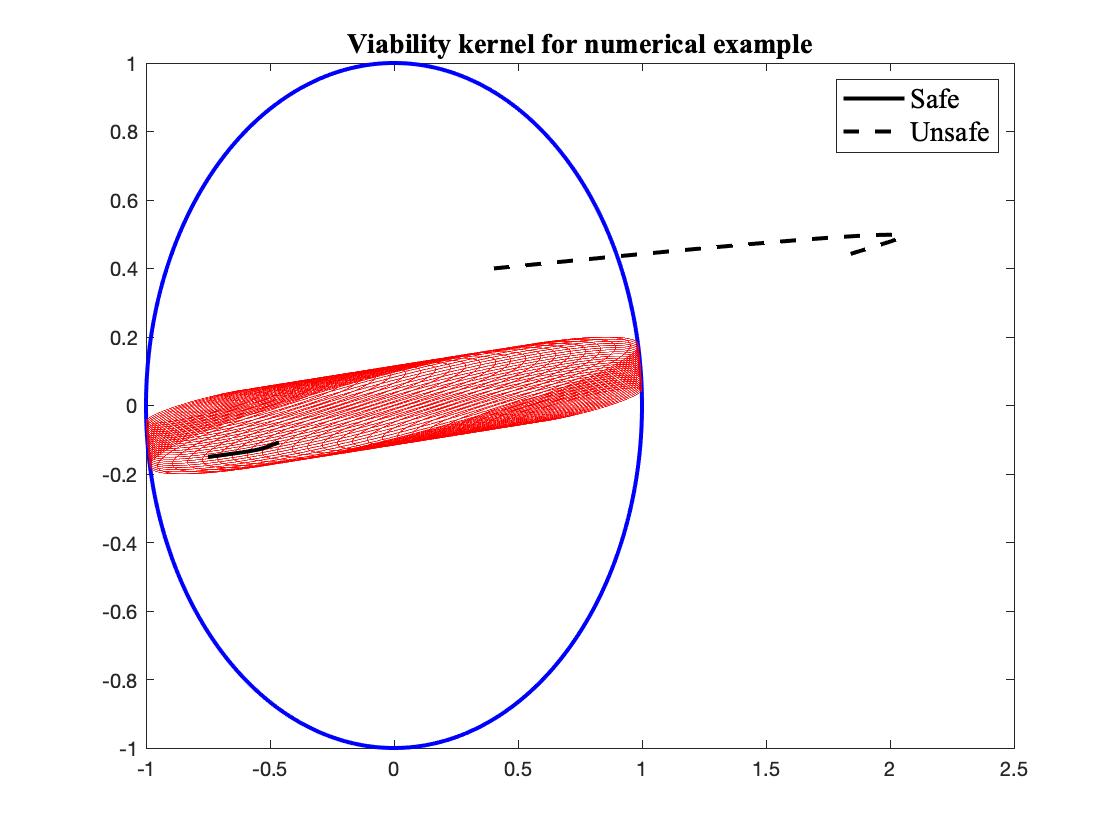} &
\includegraphics[width=3in]{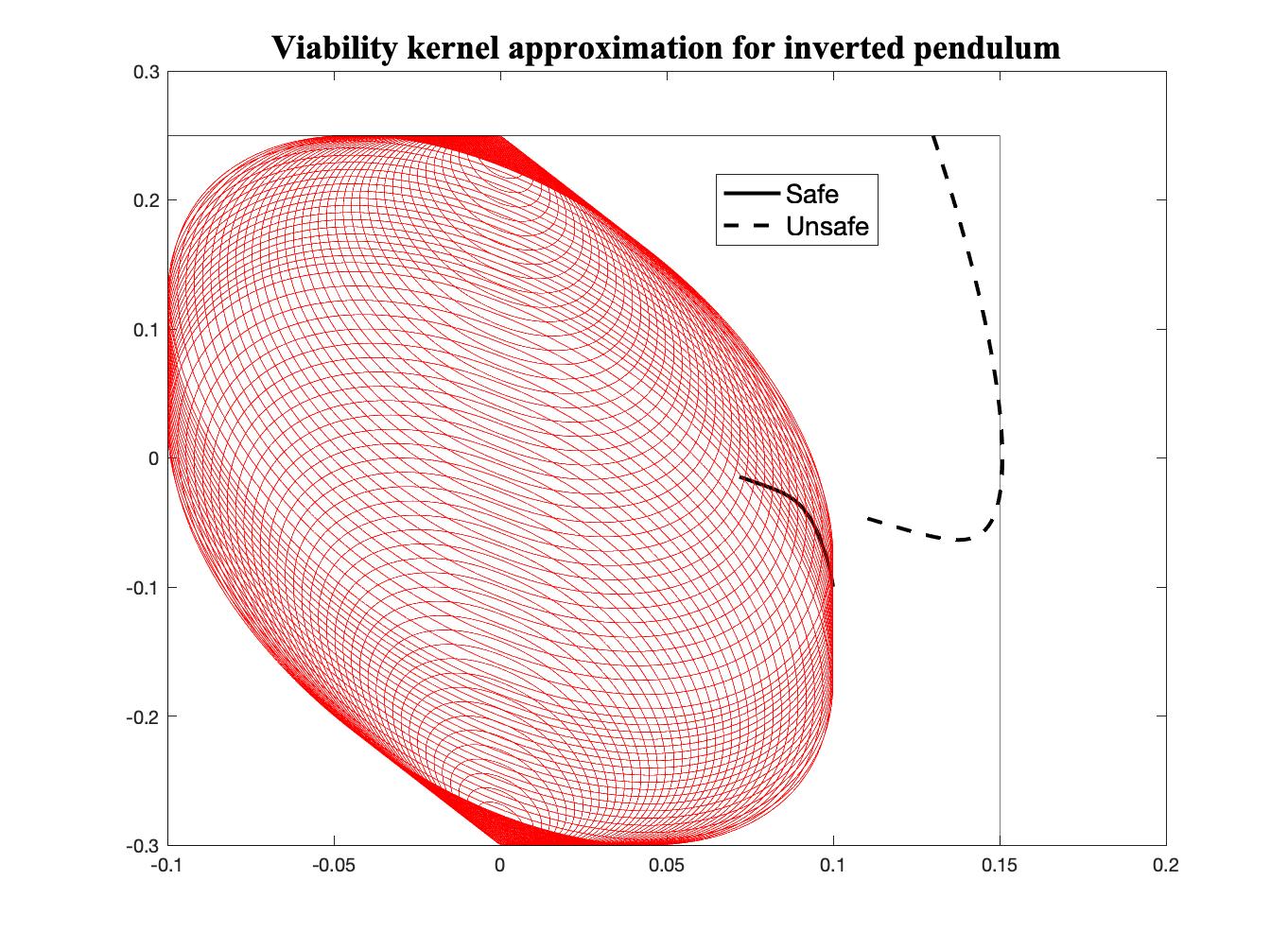} \\
\mbox{(a)} & \mbox{(b)}
\end{array}$
\caption{Illustration of our CBF framework. (a) Numerical example showed in (\ref{eq:numerical-example}). A trajectory originating inside the shaded viability kernel remains inside the safe region (outer blue circle), while a trajectory originating outside the viability kernel but inside the safe region exits the safe region. (b) Inverted pendulum with dynamics (\ref{eq:inverted-pendulum}). The trajectory originating inside the viability kernel remains safe, while the trajectory originating outside the viability kernel exits the safe region.}
\label{fig:simulation}
\end{figure*}
    
    \subsection{Convex Unsafe Region}
    \label{subsec:convex-unsafe}
    In what follows, we consider a special case where the viability kernel and $\mathcal{C}$ are identical. This is the case of controllable linear systems $\dot{x}(t) = Fx(t) + Gu(t)$ where the safe region is given by $\mathcal{C} = \mathbb{R}^{n} \setminus \bigcup_{i=1}^{m}{\{h_{i}(x) \leq 0\}}$, where each $h_{i}$ is a convex function, and the distance $d(\{h_{i}(x) \leq 0\}, \{h_{j} \leq 0\}) > \delta$ for all $i,j$ and some $\delta > 0$. We observe that this problem formulation has been considered, e.g., in \cite{khansari2012dynamical}, and safe control strategies have been proposed. This section presents a CBF-based framework for convex obstacle avoidance.
    
    \begin{Theorem}
    \label{theorem:convex-safe}
    Under the conditions described above, the viability kernel is equal to $\mathcal{C}$, and there exists a CBF-based policy that ensures safety.
    \end{Theorem}
    
    \begin{proof}
   We will prove that, for any $\epsilon > 0$, there exists a control policy such that $\bigcap_{i=1}^{m}{\{h_{i}(x) > \epsilon\}}$ is in the viability kernel. The proof is by construction. We assume that $\epsilon < \delta$. As a first step, for each $i$, we cover the set $\{x : h_{i}(x) \in (\epsilon/2, \epsilon)\}$ with a collection of  balls denoted $\{U_{ij} : j \in \mathcal{I}_{i}\}$ where $\mathcal{I}_{i}$ is an index set. For each $U_{ij}$, we let $\overline{h}_{ij}(x) = a_{ij}^{T}x - b_{i}$ be such that $\{\overline{h}_{ij}(x) \geq 0\}$ is a separating hyperplane between $U_{ij}$ and $\{h_{i}(x) < \epsilon /2\}$. 
   
   We consider each set $U_{ij}$. In the case where $a_{ij}^{T}G \neq 0$, our control policy follows the CBF constraint (\ref{eq:CBF-def}) whenever $x \in U_{ij}$. When $a_{ij}^{T}G = 0$, we have that $a_{ij}^{T}F^{r}G \neq 0$ for some $r \in \{1,\ldots,n-1\}$ since $(F,G)$ is controllable. We then define an HOCBF as
   \begin{IEEEeqnarray*}{rCl}
   \overline{h}_{ij}^{0}(x) &=& \overline{h}_{ij}(x) \\
   \overline{h}_{ij}^{1}(x) &=& a_{ij}^{T}Fx + k_{0}\overline{h}_{ij}^{0}(x) \\
   & \vdots & \\
   \overline{h}_{ij}^{r-1}(x) &=& a_{ij}^{T}F^{r-1}x + \sum_{l = 1}^{r-2}{\left(\prod_{s=l}^{r-2}{k_{s}}\right)a_{ij}^{T}F^{l}x} \\
   && + \left(\prod_{s=0}^{r-2}{k_{s}}\right)\overline{h}_{ij}^{0}(x)
   \end{IEEEeqnarray*}
   for $k_{0},\ldots,k_{r-2} > 0$. In the limit as $k_{l} \rightarrow 0$ (for $l > 0$) and $k_{0} \rightarrow \infty$, the set $$C_{ij} = \bigcap_{l=0}^{r-1}{\{\overline{h}_{ij}^{l}(x) \geq 0\}}$$ approaches $\{\overline{h}_{ij}(x) \geq 0\}$.
   
   Hence, by following this control policy, the set $\bigcap_{i=1}^{m}{\{h_{i}(x) > \epsilon\}}$ is positive invariant, implying the viability kernel is equal to $\mathcal{C}$.
    \end{proof}
\section{Simulation}
\label{sec:simulation}
We evaluated our approach through simulation of two example systems. First, we examined an unstable linear system with dynamics
\begin{equation}
\label{eq:numerical-example}
\dot{x}(t) = \left(
\begin{array}{cc}
1 & 0 \\
-1 & 4
\end{array}
\right)x(t) + \left(
\begin{array}{c}
1\\
0
\end{array}
\right)u(t)
\end{equation}
As the safe region, we chose $\{h(x) \geq 0\}$ where $h(x) = 1 - x^{T}x$, i.e., the unit disc. We first approximated the viability kernel using the approach of Section \ref{subsec:compact}, using the stabilizing feedback controller $u(t) = -Kx(t)$ where $K = (8 \ -30)$. The approximate viability kernel is shown as the red shaded in region in Fig. \ref{fig:simulation}(a). To evaluate our approximation of the viability kernel, we simulated trajectories beginning inside and outside the kernel at the points $(-0.75 \ -0.15)^{T}$ and $(-0.4 \ -0.4)^{T}$, respectively. For both points, we used a CBF-based control policy with CBF $\overline{h} = k - (x-x_{0})^{T}P(x-x_{0})$, where $k = 1.1575$, $x_{0} = (0.1378 \ 0)$, and 
\begin{displaymath}
P = \left(
\begin{array}{cc}
6.23 & -26.7 \\
-26.7 & 146.7
\end{array}
\right)
\end{displaymath}
As shown in Fig. \ref{fig:simulation}(a), the trajectory originating inside the viability kernel remains safe, whereas the trajectory originating at the point outside the viability kernel leaves the safe region.

We next evaluated our approach on an inverted pendulum with dynamics
\begin{equation}
\label{eq:inverted-pendulum}
\dot{x}(t) = 
\left(
\begin{array}{cc}
0 & 1 \\
1 & 0
\end{array}
\right)x(t) + \left(
\begin{array}{c}
0 \\
1
\end{array}
\right)u(t)
\end{equation}
and the safe region is defined by $[-0.1, 0.15] \times [-0.3, 0.25]$. This region can be viewed as the intersection of four half-plane constraints. We approximated the viability kernel of Section \ref{subsec:compact} with feedback controller $u(t) = -Kx(t)$ where $K = (3 \ 3)$. We simulated trajectories beginning inside and outside the kernel at the points $(0.1 \ -0.1)^{T}$ and $(0.13 \ 0.25)^{T}$, respectively. For both initial conditions, we used a CBF-based control policy with CBF $\overline{h} = k-(x-x_{0})^{T}P(x-x_{0})$, where $k = 0.01$, $x_{0} = 0$, and 
\begin{displaymath}
P = \left(
\begin{array}{cc}
1.25 & 0.25 \\
0.25 & 0.25
\end{array}
\right)
\end{displaymath}
The trajectory originating within the viability kernel remains in the safe region for all time, while the trajectory originating outside the viability kernel leaves the safe region returning to the origin.
\section{Conclusions and Future Work}
\label{sec:conclusion}
This paper presented a framework for verification and synthesis of control barrier functions for safe control. We mapped the  conditions for safe CBFs and HOCBFs to the existence of solutions to a system of polynomial equations. Using the Positivstellensatz, we proved that these conditions are equivalent to sum-of-squares optimization problems, which can be solved via semidefinite programming. We proposed CBF construction heuristics, including an alternating-descent semidefinite programming approach, a Lyapunov function-based method, and a CBF construction for convex unsafe regions. Our results were verified via numerical study, which showed that our approach distinguished between safe and unsafe states of an unstable linear system.

\bibliographystyle{IEEEtran}
\bibliography{CDC2021}


\end{document}